\documentclass[runningheads]{llncs}
\usepackage[T1]{fontenc}
\usepackage{amsmath,amssymb,amstext}
\usepackage{xspace}
\usepackage{enumerate}
\usepackage{color}
\usepackage{bbding}

\spnewtheorem{thm}[theorem]{Theorem}{\bfseries}{\slshape}
\spnewtheorem{dfn}[theorem]{Definition}{\bfseries}{\rmfamily}
\spnewtheorem{lem}[theorem]{Lemma}{\bfseries}{\slshape}
\spnewtheorem{conv}[theorem]{Convention}{\bfseries}{\slshape}
\spnewtheorem{cor}[theorem]{Corollary}{\bfseries}{\slshape}
\spnewtheorem{infer}[theorem]{Inference System}{\bfseries}{\rmfamily}
\spnewtheorem{exam}[theorem]{Example}{\bfseries}{\rmfamily}

\ifx \@ifpdf \@empty
        \newif\ifpdf
        \ifx\pdfoutput\undefined\pdffalse\else\pdftrue\fi
\else
        \usepackage{ifpdf}
\fi
\ifpdf
  \RequirePackage[pdftex]{graphicx}
\else
  \RequirePackage[dvips]{graphicx}
\fi

\newcommand{\acr}[1]{\scalebox{.955}{#1}}

\newcommand{\wwo}[2]{\relax\setbox0=\hbox{#1\relax}%
  \setbox1=\hbox{#2\relax}%
  \kern0.5\wd0 \kern-0.5\wd1 \hbox{#2\relax}\kern-0.5\wd1 \kern0.5\wd0}

\newcommand{\mgu}{\mathrm{mgu}}
\newcommand{\eq}{\approx}
\newcommand{\noteq}{\not\approx}

\newcommand{\ulss}{\mathrm{ss}}
\newcommand{\lss}{\mathrm{lss}}
\newcommand{\ults}{\mathrm{ts}}
\newcommand{\lts}{\mathrm{lts}}
\newcommand{\rrm}{\mathrm{rm}}
\newcommand{\nm}{\mathrm{nm}}
\newcommand{\ngt}{\mathrel{{\succ}\!\!\!{\succ}}}
\newcommand{\nlt}{\mathrel{{\prec}\!\!\!{\prec}}}
\newcommand{\succL}{\succ_{\scriptscriptstyle\mathrm{L}}}
\newcommand{\succC}{\succ_{\scriptscriptstyle\mathrm{C}}}
\newcommand{\succCL}{\succ_{\scriptscriptstyle\mathrm{Clo}}}

\newcommand{\precC}{\prec_{\scriptscriptstyle\mathrm{C}}}

\newcommand{\preceqC}{\preceq_{\scriptscriptstyle\mathrm{C}}}

\newcommand{\mul}{\mathrm{mul}}
\newcommand{\lex}{\mathrm{lex}}
\newcommand{\llor}{\;\lor\;}
\newcommand{\concl}{\mathrm{concl}}
\newcommand{\vars}{\mathrm{vars}}

\newcommand{\ie}{i.e.}

\newcommand{\wrt}{w.r.t.}
\newcommand{\DER}{\acr{DER}}
\newcommand{\Red}{\mathit{Red}}
\newcommand{\RedI}{\Red_{\mathrm{I}}}
\newcommand{\RedC}{\Red_{\mathrm{C}}}
\newcommand{\gnd}{{\mathcal{G}}}

\newcommand{\looM}{\looseness=-1}

\begin{document}
\title{On the (In-)Completeness of Destructive Equality Resolution in the Superposition Calculus\\[1.5ex]
{\large Technical Report}}
\author{Uwe Waldmann\orcidID{0000-0002-0676-7195}}
\institute{MPI for Informatics, Saarland Informatics Campus,
  Saarbr\"ucken, Germany \\
  \email{uwe@mpi-inf.mpg.de}}

\maketitle

\begin{abstract}
Bachmair's and Ganzinger's abstract redundancy concept for the
Superposition Calculus justifies almost all operations that are
used in superposition provers to delete or simplify clauses,
and thus to keep the clause set manageable.
Typical examples are tautology deletion, subsumption deletion,
and demodulation,
and with a more refined definition of redundancy
joinability and connectedness can be covered as well.
The notable exception is \textsc{Destructive Equality Resolution},
that is,
the replacement of a clause $x \not\approx t \lor C$
with $x \notin \vars(t)$ by $C\{x \mapsto t\}$.
This operation is implemented in state-of-the-art provers,
and it is clearly useful in practice,
but little is known about how it affects refutational completeness.
We demonstrate on the one hand that the naive addition of
\textsc{Destructive Equality Resolution} to the standard abstract redundancy
concept renders the calculus refutationally incomplete.
On the other hand, we present several restricted variants of
the Superposition Calculus that are refutationally complete even with
\textsc{Destructive Equality Resolution}.
\end{abstract}

\section{Introduction}
\label{sect:intro}%

Bachmair's and Ganzinger's
Superposition Calculus~\cite{BachmairGanzinger1994}
comes with an abstract redundancy concept
that describes under which circumstances clauses can be simplified away
or deleted during a saturation
without destroying the refutational completeness of
the calculus.
Typical concrete simplification and deletion techniques
that are justified by the abstract redundancy concept
are tautology deletion, subsumption deletion,
and demodulation, and with a more refined definition of
redundancy~(Duarte and Korovin \cite{DuarteKorovin2022IJCAR})
joinability and connectedness can be covered as well.

There is one simplification technique left that is
not justified by Bachmair's and Ganzinger's redundancy criterion,
namely \textsc{Destructive Equality Resolution} (\DER{}), that is,
the replacement of a clause $x \not\approx t \lor C$
with $x \notin \vars(t)$ by $C\{x \mapsto t\}$.
This operation is for instance implemented in the E prover
(Schulz \cite{Schulz2002}),
and it has been shown to be useful in practice:
It increases the number of problems that E can solve
and it also reduces E's runtime per solved problem.
The question how it affects the refutational completeness
of the calculus, both in theory and in practice,
has been open, though (except for the special case that
$t$ is also a variable, where \DER{} is equivalent
to selecting the literal $x \noteq t$ so that
\textsc{Equality Resolution} becomes the only possible inference with this clause).

In this paper we demonstrate on the one hand that the naive addition of
\DER{} to the standard abstract redundancy
concept renders the calculus refutationally incomplete.
On the other hand, we present several restricted variants of
the Superposition Calculus that are refutationally complete even with
\DER{}.

A shorter version of this report will appear in the proceedings
of \acr{IJCAR} 2024~\cite{Waldmann-DER-2024IJCAR}.

\section{Preliminaries}
\label{sect:prelim}%

\subsubsection{Basic Notions.}

We refer to (Baader and Nipkow \cite{BaaderNipkow1998})
for basic notations and results
on orderings, multiset operations, and term rewriting.

We use standard set operation symbols like $\cup$ and $\in$
and curly braces
also for finite multisets.
The union $S \cup S'$ of the multisets $S$ and $S'$ over some set $M$
is defined by $(S \cup S')(x) = S(x) + S'(x)$ for every $x \in M$.
The multiset $S$ is a submultiset of $S'$ if $S(x) \leq S'(x)$
for every $x \in M$.

We write a unification problem
in the form $u_1{\doteq}v_1,\dots u_n{\doteq}v_n$
and denote the most general unifier of such a unification problem
by $\mgu(u_1{\doteq}v_1,\dots u_n{\doteq}v_n)$.
Without loss of generality we assume that all most general
unifiers that we consider are idempotent.
Note that if $\sigma$ is an idempotent most general unifier
and $\theta$ is a unifier then
$\theta \circ \sigma = \theta$.

A clause is a finite multiset of equational literals
$s \eq t$ or $s \noteq t$,
written as a disjunction.
The empty clause is denoted by $\bot$.
We call a literal $L$ in a clause $C \lor L$ maximal
\wrt~a strict literal ordering,
if there is no literal in $C$ that is larger than $L$;
we call it strictly maximal,
if there is no literal in $C$ that is larger than or equal to $L$.

We write a rewrite rule as $u \to v$.
Semantically, a rule $u \to v$ is equivalent to an equation $u \eq v$.
If $R$ is a rewrite system, that is, a set of rewrite rules,
we write $s \to_R t$ to indicate that the term $s$ can be reduced to
the term $t$ by applying a rule from $R$.
A rewrite system is called left-reduced, if there is no rule
$u \to v \in R$ such that $u$ can be reduced by a rule from
$R \setminus \{u \to v\}$.

\subsubsection{The Superposition Calculus.}

We summarize the key elements of Bachmair's and Ganzinger's
Superposition Calculus~\cite{BachmairGanzinger1994}.

Let $\succ$ be a reduction ordering that is total on ground terms.
We extend $\succ$ to an ordering on literals,
denoted by $\succL$,\footnote{%
  There are several equivalent ways to define $\succL$.}
by mapping
positive literals $s \eq t$ to multisets $\{s,t\}$
and negative literals $s \noteq t$ to multisets $\{s,s,t,t\}$
and by comparing the resulting multisets using the multiset extension
of $\succ$.
We extend the literal ordering $\succL$ to an ordering on clauses,
denoted by $\succC$,
by comparing the multisets of literals in these clauses
using the multiset extension of $\succL$.

The inference system of the Superposition Calculus
consists of the rules
\textsc{Superposition},
\textsc{Equality Resolution}, and
either \textsc{Ordered Factoring} and \textsc{Merging Paramodulation} or
\textsc{Equality Factoring}.

\medskip
\noindent\begin{tabular}{@{}p{14em}@{\quad}l@{}}
  \textsc{Superposition:}
&
  $\displaystyle{\frac
  {D' \lor {t \eq t'} \qquad C' \lor L[{u}]}
  {(D' \lor C' \lor L[{t'}])\sigma}
  }$
\end{tabular}\par\nobreak\noindent\begin{itemize}\item[]
  where $u$ is not a variable;
  $\sigma = \mgu(t {\doteq} u)$;
  $(C' \lor L[{u}])\sigma \not\preceqC (D' \lor {t \eq t'})\sigma$;
  $(t \eq t')\sigma$ is strictly maximal
  in $(D' \lor {t \eq t'})\sigma$;
  either $L[{u}]$ is a positive literal $s[{u}] \eq s'$
  and $L[{u}]\sigma$ is strictly maximal
  in $(C' \lor L[{u}])\sigma$,
  or $L[{u}]$ is a negative literal $s[{u}] \noteq s'$
  and $L[{u}]\sigma$ is maximal
  in $(C' \lor L[{u}])\sigma$;
  $t\sigma \not\preceq t'\sigma$;
  and $s\sigma \not\preceq s'\sigma$.
\end{itemize}

\medskip
\noindent\begin{tabular}{@{}p{14em}@{\quad}l@{}}
  \textsc{Equality Resolution:}
&
  $\displaystyle{\frac
  {C' \lor {s \noteq s'}}
  {C'\sigma}
  }$
\end{tabular}\par\nobreak\noindent\begin{itemize}\item[]
  where $\sigma = \mgu(s {\doteq} s')$
  and $(s \noteq s')\sigma$ is maximal in $(C' \lor {s \noteq s'})\sigma$.
\end{itemize}

\medskip
\noindent\begin{tabular}{@{}p{14em}@{\quad}l@{}}
  \textsc{Ordered Factoring:}
&
  $\displaystyle{\frac
  {C' \lor L' \lor L}
  {(C' \lor L')\sigma}
  }$
\end{tabular}\par\nobreak\noindent\begin{itemize}\item[]
  where $L$ and $L'$ are positive literals;
  $\sigma = \mgu(L {\doteq} L')$;
  and
  $L\sigma$ is maximal in $(C' \lor L' \lor L)\sigma$.
\end{itemize}

\medskip
\noindent\begin{tabular}{@{}p{14em}@{\quad}l@{}}
  \textsc{Merging Paramodulation:}
&
  $\displaystyle{\frac
  {D' \lor {t \eq t'} \qquad C' \lor r \eq r' \lor s \eq s'[u]}
  {(D' \lor C' \lor r \eq r' \lor s \eq s'[t'])\sigma}
  }$
\end{tabular}\par\nobreak\noindent\begin{itemize}\item[]
  where $u$ is not a variable;
  $\sigma = \mgu(r {\doteq} s, t {\doteq} u)$;
  $(C' \lor r \eq r' \lor s \eq s'[u])\sigma \not\preceqC (D' \lor {t \eq t'})\sigma$;
  $t\sigma \not\preceq t'\sigma$;
  $s\sigma \not\preceq s'\sigma$;
  $(t \eq t')\sigma$ is strictly maximal
  in $(D' \lor {t \eq t'})\sigma$;
  and
  $(s \eq s')\sigma$ is strictly maximal in
  $(C' \lor r \eq r' \lor s \eq s'[u])\sigma$.
\end{itemize}

\medskip
\noindent\begin{tabular}{@{}p{14em}@{\quad}l@{}}
  \textsc{Equality Factoring:}
&
  $\displaystyle{\frac
  {C' \lor {r} \eq r' \lor {s} \eq s'}
  {(C' \lor s' \noteq r' \lor r \eq r')\sigma}
  }$
\end{tabular}\par\nobreak\noindent\begin{itemize}\item[]
  where $\sigma = \mgu(s {\doteq} r)$;
  $s\sigma \not\preceq s'\sigma$;
  and
  $({s} \eq s')\sigma$ is maximal in
  $(C' \lor {r} \eq r' \lor {s} \eq s')\sigma$.
\end{itemize}

The ordering restrictions can be overridden using \emph{selection functions}
that determine for each clause a subset of the negative literals
that are available for inferences.
For simplicity, we leave out this refinement in the rest of this
paper.
We emphasize, however, that all results that we present here
hold also in the presence of selection functions;
the required modifications of the proofs are straightforward.

A ground clause $C$ is called
(classically) redundant \wrt~a set of ground clauses $N$,
if it follows
from clauses in $N$
that are smaller than $C$ \wrt~$\succC$.
A clause is called (classically) redundant \wrt~a set of clauses $N$,
if all its ground instances are redundant \wrt~the set of ground instances
of clauses in $N$.\footnote{%
  Note that ``redundancy'' is called ``compositeness''
  in Bachmair and Ganzinger's 
  \emph{J.~Log.\ Comput.}\ article~\cite{BachmairGanzinger1994}.
  In later papers the standard terminology has changed.}
A ground inference with conclusion $C_0$ and right (or only)
premise $C$ is called 
redundant \wrt~a set of ground clauses $N$,
if one of its premises is redundant \wrt~$N$,
or if
$C_0$ follows
from clauses in $N$
that are smaller than $C$.
An inference is called redundant \wrt~a set of clauses $N$,
if all its ground instances are redundant \wrt~the set of ground instances
of clauses in $N$.

Redundancy of clauses and inferences as defined above is a
redundancy criterion in the sense of
(Waldmann et al.~\cite{WaldmannTourretRobillardBlanchette2022}).
It justifies typical deletion and simplification techniques such as
the deletion of tautological clauses,
subsumption deletion
(\ie, the deletion of a clause $C\sigma \lor D$ in the presence
of a clause $C$)
or demodulation
(\ie, the replacement of a clause $C[s\sigma]$
by $C[t\sigma]$ in the presence of a unit clause $s \eq t$,
provided that $s\sigma \succ t\sigma$).

\section{Incompleteness}
\label{sect:incompl}%

There are two special cases
where \textsc{Destructive Equality Resolution} (\DER{}) is justified
by the classical redundancy criterion:
First, if $t$ is the smallest constant in the signature,
then every ground instance $(x \noteq t \lor C)\theta$
follows from the smaller ground instance $C\{x \mapsto t\}\theta$.
Second, if $t$ is another variable $y$,
then every ground instance $(x \noteq y \lor C)\theta$
follows from the smaller ground instance
$C\{x \mapsto y\}\{y \mapsto s\}\theta$,
where $s$ is the smaller of $x\theta$ and $y\theta$.

But it is easy to see that this does not work
in general:
Let $\succ$ be a Knuth-Bendix ordering with weights
$w(f) = w(b) = 2$,
$w(c) = w(d) = 1$,
$w(z) = 1$ for all variables $z$,
and let $C$ be the clause $x \noteq b \llor f(x) \eq d$,
Then \DER{} applied to $C$ yields
$D\ = \ f(b) \eq d$.
Now consider the substitution $\theta = \{x \mapsto c\}$.
The ground instance $C\theta\ = \ {c \noteq b} \llor {f(c) \eq d}$
is a logical consequence of $D$, but since it is smaller than $D$ itself,
$D$ makes neither $C\theta$ nor $C$ redundant.

Moreover, the following example demonstrates that
the Superposition Calculus becomes indeed incomplete,
if we add \DER{} as a simplification rule,
\ie, if we extend the definition of redundancy in such a way
that the conclusion of \DER{} renders
the premise redundant.

\begin{exam}
\label{ex:incompleteness}%
Let $\succ$ be a Knuth-Bendix ordering with weights
$w(f) = 4$,
$w(g) = 3$,
$w(b) = 4$,
$w(b') = 2$,
$w(c) = w(c') = w(d) = 1$,
$w(z) = 1$ for all variables $z$,
and let $N$ be the set of clauses
\[\begin{array}{@{}l@{}}
  C_1 \ =\ \underline{f(x,d)} \eq x
\\[1ex]
  C_2 \ =\ \underline{f(x,y)} \noteq b \llor g(x) \eq d
\\[1ex]
  C_3 \ =\ b' \eq c' \llor \underline{b} \eq c
\\[1ex]
  C_4 \ =\ \underline{g(b')} \noteq g(c')
\\[1ex]
  C_5 \ =\ \underline{g(c)} \noteq d
\end{array}\]
where all the maximal terms in maximal literals are underlined.

At this point, neither demodulation nor subsumption is possible.
The only inference that must be performed is \textsc{Superposition}
between $C_1$ and $C_2$,
yielding
\[\begin{array}{@{}l@{}}
  C_6 \ =\ \rlap{$x \noteq b \llor g(x) \eq d$}\hphantom{f(x,y) \noteq b \llor g(x) \eq d}
\end{array}\]
and by using \DER{}, $C_6$ is replaced by
\[\begin{array}{@{}l@{}}
  C_7 \ =\ \rlap{$g(b) \eq d$}\hphantom{f(x,y) \noteq b \llor g(x) \eq d}
\end{array}\]

We could now continue with a \textsc{Superposition}
between $C_3$ and $C_7$,
followed by a \textsc{Superposition} with $C_5$,
followed by \textsc{Equality Resolution},
and obtain
\[\begin{array}{@{}l@{}}
  C_8 \ =\ \rlap{$b' \eq c'$}\hphantom{f(x,y) \noteq b \llor g(x) \eq d}
\end{array}\]
from which we can derive the empty clause by
\textsc{Superposition} with $C_4$
and once more by \textsc{Equality Resolution}.
However, clause $C_7$ is in fact redundant:
The ground clauses $C_3$ and $C_4$ imply $b \eq c$;
therefore $C_7$ follows from $C_3$, $C_4$, and
the ground instances
\[\begin{array}{@{}r@{}l@{\qquad}l@{}}
  C_1\{x \mapsto c\} & {} \ =\ f(c,d) \eq c
\\[1ex]
  C_2\{x \mapsto c,\,y \mapsto d\} & {} \ =\ f(c,d) \noteq b \llor g(c) \eq d
\end{array}\]
Because all terms in these clauses are smaller than the maximal term
$g(b)$ of $C_7$,
all the clauses are smaller than $C_7$.
Since $C_7$ is redundant, we are allowed to delete it,
and then no further inferences are possible anymore.
Therefore the clause set $N = \{C_1,\ldots,C_5\}$ is saturated,
even though it is inconsistent
and does not contain the empty clause,
which implies that the calculus is not refutationally complete anymore.
\end{exam}

\section{Completeness, Part I: The Horn Case}
\label{sect:complh}%

\subsection{The Idea}

On the one hand,
Example \ref{ex:incompleteness} demonstrates
that we cannot simply extend the standard redundancy criterion
of the Superposition Calculus
with \DER{} without destroying
refutational completeness, and that this holds even
if we impose a particular strategy on simplification steps
(say, that simplifications must be performed eagerly and that
demodulation and subsumption have a higher precedence
than \DER{}).
On the other hand, Example \ref{ex:incompleteness} is of course
highly unrealistic:
Even though clause $C_7$ is redundant \wrt~the clauses
$C_1$, $C_2$, $C_3$, and $C_4$,
no reasonable superposition prover would ever detect this -- in particular,
since doing so would require to invent
the instance $C_2\{x \mapsto c,\,y \mapsto d\}$ of $C_2$,
which is not in any way syntactically related to $C_7$.\footnote{%
  In fact, a prover might use SMT-style heuristic grounding of non-ground
  clauses, but then finding the contradiction turns out to be easier than
  proving the redundancy of $C_7$.}

This raises the question whether \DER{}
still destroys refutational completeness when we restrict
the other deletion and simplification techniques to those that are typically
implemented in superposition provers, such as
tautology detection, demodulation, or subsumption.
Are there alternative redundancy criteria
that are refutationally complete
together with the Superposition Calculus
and that justify \DER{}
as well as
(all/most) commonly implemented deletion and simplification techniques?
Given the usual structure of the inductive completeness proofs for
saturation calculi,
developing such a redundancy criterion would mean in particular
to find a suitable clause ordering with respect to which certain clauses
have to be smaller than others.
The following example illustrates a fundamental problem
that we have to deal with:

\begin{exam}
Let $\succ$ be a Knuth-Bendix ordering with weights
$w(f) = w(g) = w(h) = w(c) = 1$, $w(b) = 2$,
$w(z) = 1$ for all variables $z$.
Consider the following set of clauses:
\[
\begin{array}{@{}l@{\qquad\qquad}l@{}}
D_1 \ =\  h(x) \eq x
&
C_1 \ =\  h(x) \noteq b \lor f(g(x)) \eq c
\\[1ex]
&
C_2 \ =\  x \noteq b \lor f(g(x)) \eq c
\\[1ex]
D_3 \ =\  h(c) \noteq b \lor g(b) \eq g(c)
&
C_3 \ =\  f(g(b)) \eq c
\\[1ex]
&
C_4 \ =\  h(c) \noteq b \lor f(g(c)) \eq c
\end{array}
\]

Demodulation of $C_1$ using $D_1$
yields $C_2$, and
if we want Demodulation to be a simplification,
then every ground instance $C_1\theta$ should be larger than
the corresponding ground instance $C_2\theta$
in the clause ordering.

\DER{} of $C_2$ yields $C_3$,
and if we want \DER{} to be a simplification,
then every ground instance $C_2\theta$ should be larger than $C_3\theta = C_3$.

A \textsc{Superposition} inference between $D_3$ and $C_3$ yields $C_4$.
The inductive completeness proof for the calculus relies on the fact
that the conclusion of an inference is smaller than the largest premise,
so $C_3$ should be larger than $C_4$.

By transitivity we obtain that every ground instance $C_1\theta$ should be
larger than $C_4$ in the clause ordering.
The clause $C_4$, however,
\emph{is} a ground instance of $C_1$,
which is clearly a contradiction.
\end{exam}

On the other hand, a closer inspection reveals that, depending on the limit
rewrite system $R_*$ that is produced in the completeness proof
for the Superposition Calculus,
the \textsc{Superposition} inference between $D_3$ and $C_3$
is only needed,
when $D_3$ produces the rewrite rule $g(b) \to g(c) \in R_*$,
and that the only critical case for \DER{}
is the one where $b$ can be reduced by some rule in $R_*$.
Since the limit rewrite system $R_*$ is by construction left-reduced,
these two conditions are mutually exclusive.
This observation indicates that we might be able to find a suitable
clause ordering if we choose it depending on $R_*$.

\subsection{Ground Case}

\subsubsection{The Normalization Closure Ordering.}

Let $\succ$ be a reduction ordering that is total on ground terms.
Let $R$ be a left-reduced ground rewrite system contained in $\succ$.

For technical reasons that will become clear later,
we design our ground superposition calculus in such a way
that it operates on ground closures $(C\cdot\theta)$.
Logically, a ground closure $(C\cdot\theta)$ is equivalent to
a ground instance $C\theta$, but an ordering may treat two
closures that represent the same ground instance in different ways.
We consider closures up to $\alpha$-renaming
and ignore the behavior of $\theta$ on variables that do not occur in $C$,
that is, we treat closures $(C_1\cdot\theta_1)$ and $(C_2\cdot\theta_2)$
as equal whenever
$C_1$ and $C_2$ are equal up to bijective variable renaming
and $C_1\theta_1 = C_2\theta_2$.
We also identify $(\bot\cdot\theta)$ and $\bot$.

Intuitively, in order to compare ground closures $C\cdot\theta$,
we normalize all terms occurring in $C\theta$ with $R$,
we compute the multiset of all the redexes occurring during
the normalization and all the resulting normal forms,
and we compare these multisets using the multiset extension of $\succ$.
Since we would like to give redexes and normal forms in negative literals
a slightly larger weight than redexes and normal forms in positive literals,
and redexes in positive literals below the top
a slightly larger weight than redexes at the top,
we combine each of these terms with a label
($0$ for positive at the top, $1$ for positive below the top,
$2$ for negative).
Moreover, whenever some term $t$ occurs several times in $C$ as a subterm,
we want to count the redexes resulting from the normalization of $t\theta$
only once (with the maximum of the labels).
The reason for this is that \DER{} can produce several copies of
the same term $t$ in a clause if the variable to be eliminated
occurs several times in the clause;
by counting all redexes stemming from $t$ only once, we ensure that
this does not increase the total number of redexes.
Formally, we first compute the set (not multiset!)
of all subterms $t$ of $C$, so that duplicates are deleted,
and then compute the multiset of redexes for all terms $t\theta$
(and analogously for terms occurring at the top of a literal).

\begin{dfn}
\label{dfn:ss-ts-lss-rm}%
We define the subterm sets $\ulss^{+}_{>\epsilon}(C)$ and $\ulss^{-}(C)$
and the topterm sets $\ults^{+}(C)$ and $\ults^{-}(C)$
of a clause $C$ by
\[\begin{array}{@{}r@{}l@{}}
  \ulss^{-}(C) = {} & \{\, t \mid C = C' \lor s[t]_p \noteq s' \,\} \\[1ex]
  \ulss^{+}_{>\epsilon}(C) = {} & \{\, t \mid C = C' \lor s[t]_p \eq s',\,p > \epsilon \,\} \\[1ex]
  \ults^{-}(C) = {} & \{\, t \mid C = C' \lor t \noteq t' \,\} \\[1ex]
  \ults^{+}(C) = {} & \{\, t \mid C = C' \lor t \eq t' \,\}\,.
\end{array}\]

We define the labeled subterm set $\lss(C)$ and the
labeled topterm set $\lts(C)$ of a clause $C$ by
\[\begin{array}{@{}r@{}l@{}}
  \lss(C) = {} & \{\, (t,2) \mid t \in \ulss^{-}(C) \,\} \\[2pt]
               & \qquad {} \cup \{\, (t,1) \mid t \in \ulss^{+}_{>\epsilon}(C) \setminus \ulss^{-}(C) \,\} \\[2pt]
               & \qquad {} \cup \{\, (t,0) \mid t \in \ults^{+}(C) \setminus (\ulss^{+}_{>\epsilon}(C) \cup \ulss^{-}(C)) \,\} \\[1ex]
  \lts(C) = {} & \{\, (t,2) \mid t \in \ults^{-}(C) \,\}
                 \cup \{\, (t,0) \mid t \in \ults^{+}(C) \setminus \ults^{-}(C) \,\}\,.
\end{array}\]

We define the $R$-redex multiset $\rrm_R(t,m)$
of a labeled ground term $(t,m)$ with $m \in \{0,1,2\}$ by
\[\begin{array}{@{}r@{}l@{}}
  \rrm_R(t,m) = {} & \emptyset \text{~if $t$ is $R$-irreducible;} \\[1ex]
  \rrm_R(t,m) = {} & \{(u,m)\} \cup \rrm_R(t',m) \text{~if $t \to_R t'$ using the rule $u \to v \in R$}\\[2pt]
                  & \qquad \text{at position $p$ and $p = \epsilon$ or $m > 0$;} \\[1ex]
  \rrm_R(t,m) = {} & \{(u,1)\} \cup \rrm_R(t',m) \text{~if $t \to_R t'$ using the rule $u \to v \in R$}\\[2pt]
                  & \qquad \text{at position $p$ and $p > \epsilon$ and $m = 0$.}
\end{array}\]
\end{dfn}

\begin{lem}
\label{lem:rrm-welldefined}%
For every left-reduced ground rewrite system $R$ contained in $\succ$,
$\rrm_R(t,m)$ is well-defined.
\end{lem}

\begin{proof}
We write the function $\rrm_R(t,m)$ as a ternary relation
$\rrm'_R(t,m,S)$ and show that for every $t$ and $m$
there is exactly one multiset $S$ such that $\rrm'_R(t,m,S)$ holds:
\[\begin{array}{@{}l@{}l@{}}
  \rrm'_R(t,m,\emptyset)        & \text{~~if $t$ is $R$-irreducible;} \\[1ex]
  \rrm'_R(t,m,S \cup \{(u,m)\}) & \text{~~if $t \to_R t'$ using the rule $u \to v \in R$} \\[2pt]
                                & \qquad \text{at position $p$, $p = \epsilon$ or $m > 0$, and $\rrm'(t',m,S)$.} \\[1ex]
  \rrm'_R(t,m,S \cup \{(u,1)\}) & \text{~~if $t \to_R t'$ using the rule $u \to v \in R$}\\[2pt]
                                & \qquad \text{at position $p$, $p > \epsilon$, $m = 0$, and $\rrm_R(t',m,S)$.}
\end{array}\]

For simplicity, we consider only the case $m = 0$;
the cases $m = 1$ and $m = 2$ are proved analogously.

An obvious induction over the term ordering $\succ$ shows that
for every term $t$ there is
at least one multiset $S$ such that $\rrm'_R(t,0,S)$ holds.

It remains to show that there is at most one such multiset.
For this part, we proceed as in the well-known proof of Newman's lemma
(also known as ``Diamond lemma'').
We observe first that if $t$ is $R$-reducible by a rule $u \to v$
at the position $p$, 
then $p$ determines $u$ uniquely since $R$ is ground,
and $u$ determines $v$ uniquely since $R$ is left-reduced.

Assume that there is a term $t$
such that $\rrm'_R(t,0,S_1)$ and $\rrm'_R(t,0,S_2)$ hold
for two different multisets $S_1$ and $S_2$.
By well-foundedness of $\succ$, we may assume that
$t$ is minimal w.r.t.~$\succ$ with this property.
Clearly,
$t$ must be $R$-reducible, otherwise there is exactly one multiset $S$
such that $\rrm'_R(t,0,S)$ holds, namely $S = \emptyset$.
Assume that there are positions $p_i$ $(i \in \{1,2\})$
such that $t$ can be rewritten to $t_i = t[v_i]_{p_i}$
using rules $u_i \to v_i$ at $p_i$,
that $\rrm'(t_i,0,S'_i)$,
and that $S_i = S'_i \cup \{(u_i,m_i)\}$,
where $m_i = 0$ if $p_i = \epsilon$ and $m_i = 1$ otherwise.

We know that $t \succ t_i$, and
by minimality of $t$ this implies that there is only one $S'_1$
and one $S'_2$ such that  
$\rrm'(t_i,0,S'_i)$ holds.
If $p_1$ and $p_2$ were equal, then
$u_1 = u_2$ and $m_1 = m_2$,
therefore $v_1 = v_2$ and $t_1 = t_2$,
therefore $S'_1 = S'_2$, and therefore $S_1 = S_2$, contradicting our
assumption.
Thus we know that $p_1 \not= p_2$,
and since $R$ is left-reduced and ground, this implies that
$p_1$ and $p_2$ must be parallel positions, that
neither of them can be $\epsilon$, and that $m_1 = m_2 = 1$.

Let $t_3 = t[v_1]_{p_1}[v_2]_{p_2}$.
Then $t_1 \to_R t_3$ using $u_2 \to v_2$ at $p_2$
and $t_2 \to_R t_3$ using $u_1 \to v_1$ at $p_1$.
By minimality of $t$,
there is only one $S'_3$ such that
$\rrm'(t_3,0,S'_3)$,
and furthermore
$S'_1 = S'_3 \cup \{(u_2,1)\}$
and
$S'_2 = S'_3 \cup \{(u_1,1)\}$.
But then
$S_1 = S'_1 \cup \{(u_1,1)\} = S'_3 \cup \{(u_2,1),(u_1,1)\}$
and
$S_2 = S'_2 \cup \{(u_2,1)\} = S'_3 \cup \{(u_1,1),(u_2,1)\}$,
contradicting the assumption that $S_1 \not= S_2$.
\qed
\end{proof}

\begin{dfn}
\label{dfn:nm-ngt}%
We define the $R$-normalization multiset $\nm_R(C\cdot\theta)$
of a ground closure $(C\cdot\theta)$ by
\[\begin{array}{@{}r@{}l@{}}
  \nm_R(C\cdot\theta) = {} & \,\bigcup_{(f(t_1,\dots,t_n),m) \in \lss(C)}
         \rrm_R(f(t_1\theta{\downarrow}_R,\dots,t_n\theta{\downarrow}_R),m) \\[3pt]
 & \hspace*{2em} {} \cup \bigcup_{(x,m) \in \lss(C)} \rrm_R(x\theta,m) \\[3pt]
 & \hspace*{2em} {} \cup \bigcup_{(t,m) \in \lts(C)} \{ (t\theta{\downarrow}_R,m) \}
\end{array}\]

\begin{exam}
Let $C = h(g(g(x))) \eq f(f(b))$;
let $\theta = \{x\mapsto b\}$.
Then 
$\lss(C) = \{(h(g(g(x))),0),$ $(g(g(x)),1),$ $(g(x),1),$ $(x,1),$
$(f(f(b)),0),$ $(f(b),1),$ $(b,1)\}$
and
$\lts(C) = \{(h(g(g(x))),0),$ $(f(f(b)),0)\}$.

Let $R = \{f(b) \to b,$ $g(g(b)) \to b\}$.
Then $\nm_R(C\cdot\theta) = \{(g(g(b)),1),$ $(f(b),1),$ $(f(b),0),$
$(h(b),0),$ $(b,0)\}$,
where the first element is a redex from the normalization of $g(g(x))\theta$,
the second from the normalization of $f(b)\theta$,
the third from the normalization of $f(f(b))\theta$.
The remaining elements are the normal forms
of $h(g(g(x)))\theta$ and $f(f(b))\theta$.
\end{exam}

The $R$-normalization closure ordering $\ngt_R$ compares
ground closures $(C\cdot\theta_1)$ and $(D\cdot\theta_2)$
using a lexicographic combination of three orderings:
\begin{itemize}
\item
  first, the multiset extension $(({\succ},{>})_\lex)_\mul$
  of the lexicographic combination of
  the reduction ordering $\succ$ and the ordering $>$ on natural numbers
  applied to the multisets $\nm_R(C\cdot\theta_1)$ and $\nm_R(D\cdot\theta_2)$,
\item
  second, the traditional clause ordering $\succC$ applied to
  $C\theta_1$ and $D\theta_2$,
\item
  and third, an arbitrary well-founded ordering $\succCL$ on ground closures
  that is total on ground closures
  $(C\cdot\theta_1)$ and $(D\cdot\theta_2)$
  with $C\theta_1 = D\theta_2$
  and that has the property that
  $(C\cdot\theta_1) \succCL (D\cdot\theta_2)$
  whenever $C\theta_1 = D\theta_2$ and $D$ is an instance of $C$
  but not vice versa.
\end{itemize}
\end{dfn}

\begin{lem}
\label{lem:ordering-of-instances}%
If $(C\cdot\theta)$ and $(C\sigma\cdot\theta')$
are ground closures,
such that
$C\theta = C\sigma\theta'$,
and $C$ and $C\sigma$ are not equal up to bijective renaming,
then $(C\cdot\theta) \ngt_R (C\sigma\cdot\theta')$.
\end{lem}

\begin{proof}
The $R$-normalization of $C\theta$ and $C\sigma\theta'$
yields the same redexes and normal forms.
Moreover, whenever a term $t$ occurs multiple times in $C$
(which means that the redexes stemming from these term occurrences
are counted only once),
then the terms occurring at the corresponding positions in $C\sigma$
are equal as well
(so the redexes stemming from these term occurrences
are again counted only once).
Therefore $\nm_R(C\cdot\theta) \supseteq \nm_R(C\sigma\cdot\theta')$.

If $\nm_R(C\cdot\theta) \supset \nm_R(C\sigma\cdot\theta')$, then
$\nm_R(C\cdot\theta) \mathrel{(({\succ},{>})_\lex)_\mul}
\nm_R(C\sigma\cdot\theta')$
and hence
$(C\cdot\theta) \ngt_R (C\sigma\cdot\theta')$;
otherwise
$\nm_R(C\cdot\theta) = \nm_R(C\sigma\cdot\theta')$,
$C\theta = C\sigma\theta'$,
and
$(C\cdot\theta) \succCL (C\sigma\cdot\theta')$,
hence again
$(C\cdot\theta) \ngt_R (C\sigma\cdot\theta')$.
\qed
\end{proof}

\begin{exam}
Let $C = h(f(x)) \eq f(y)$,
let $\theta' = \{x\mapsto b\}$;
let $\theta = \{x\mapsto b,$ $y\mapsto b\}$;
let $\sigma = \{y \mapsto x\}$.
Let $R = \{f(b) \to b\}$.

Then $\nm_R(C \cdot \theta) =
\{(f(b),1),$ $(f(b),0),$ $(h(b),0),$ $(b,0)\}$
and
$\nm_R(C\sigma \cdot \theta') =
\{(f(b),1),$ $(h(b),0),$ $(b,0)\}$,
and therefore
$(C\cdot\theta) \ngt_R (C\sigma\cdot\theta')$.
The subterm $f(x)$ occurs twice in $C\sigma$
(with labels $0$ and $1$),
but only once in $\lss(C\sigma)$
(with the larger of the two labels),
and the same holds for the redex $f(b)$
stemming from $f(x)\theta'$
in $\nm_R(C\sigma \cdot \theta')$.
\end{exam}

\subsubsection{Parallel Superposition.}

In the normalization closure ordering,
redexes and normal forms stemming from several occurrences
of the same term $u$ in a closure $(C \cdot \theta)$
are counted only once.
\pagebreak[3]
When we perform a \textsc{Superposition} inference,
this fact leads to a small problem:
Consider a closure $(C[u,u] \cdot \theta)$.
In the $R$-normalization multiset of this closure,
the redexes stemming from the two copies of $u\theta$
are counted only once.
Now suppose that one of the two copies of $u$ is replaced by
a smaller term $v$ in a \textsc{Superposition} inference.
The resulting closure $(C[v,u] \cdot \theta)$ should be smaller
than the original one, but it isn't: The redexes
stemming from $u\theta$ are still counted once,
and additionally, the $R$-normalization multiset
now contains the redexes stemming from $v\theta$.

There is an easy fix for this problem, though:
We have to replace the ordinary \textsc{Superposition} rule
by a \textsc{Parallel Superposition} rule, in which
\emph{all} copies of a term $u$ in a clause $C$ are replaced
whenever one copy occurs in a maximal side of a maximal literal.
Note that this is a well-known optimization that superposition provers
implement (or should implement) anyhow.

We need one further modification of the inference rule:
The side conditions of the superposition calculus
use the traditional clause ordering $\succC$,
but our completeness proof and redundancy criterion will be based on
the orderings $\ngt_R$.
The difference between these orderings
becomes relevant in particular when we consider
\textsc{(Parallel) Superposition}
inferences where the clauses overlap at the top
of a positive literal.
In this case, the $\ngt_R$-smaller of the two premises
may actually be the $\succC$-larger one.
Therefore, the usual condition
that the left premise of a \textsc{(Parallel) Superposition}
inference has to be $\succC$-minimal
has to be dropped for these inferences.

\medskip
\noindent\begin{tabular}{@{}p{14em}@{\quad}l@{}}
  \textsc{Parallel Superposition:}
&
  $\displaystyle{\frac
  {D' \lor {t \eq t'} \qquad C[u,\ldots,u]_{p_1,\dots,p_k}}
  {(D' \lor C[t',\ldots,t']_{p_1,\dots,p_k})\sigma}
  }$
\end{tabular}\par\nobreak\noindent\begin{itemize}\item[]
  where $u$ is not a variable;
  $\sigma = \mgu(t {\doteq} u)$;
  $p_1,\dots,p_k$ are all the occurrences of $u$ in $C$;
  if one of the occurrences of $u$ in $C$ is in a negative literal
  or below the top in a positive literal then
  $C\sigma \not\preceqC (D' \lor {t \eq t'})\sigma$;
  $(t \eq t')\sigma$ is strictly maximal
  in $(D' \lor {t \eq t'})\sigma$;
  either one of the occurrences of $u$ in $C$ is in a
  positive literal $L[{u}] = s[{u}] \eq s'$
  and $L[{u}]\sigma$ is strictly maximal
  in $C\sigma$,
  or one of the occurrences of $u$ in $C$ is in a
  negative literal $L[{u}] = s[{u}] \noteq s'$
  and $L[{u}]\sigma$ is maximal
  in $(C' \lor L[{u}])\sigma$;
  $t\sigma \not\preceq t'\sigma$;
  and $s\sigma \not\preceq s'\sigma$.
\end{itemize}

\subsubsection{Ground Closure Horn Superposition.}

We will show that our calculus is refutationally complete
for Horn clauses by lifting a similar result for
ground closure Horn superposition.
We emphasize that our calculus is not a basic or constraint calculus
such as
(Bachmair et al.~\cite{BachmairGanzingerLynchSnyder1995})
or (Nieuwenhuis and Rubio~\cite{nieuwenhuis-rubio-1995}).
Even though the ground version that we present here operates on closures,
it is essentially a rephrased version of the
standard ground Superposition Calculus.
This explains why we also have to consider superpositions below
variable positions.

The ground closure calculus uses the following three inference rules.
We assume that in binary inferences the variables in the premises
$(D\cdot\theta_2)$ and $(C\cdot\theta_1)$ are renamed
in such a way that $C$ and $D$ do not share variables.
We can then assume without loss of generality that the substitutions
$\theta_2$ and $\theta_1$ agree.

\pagebreak[3]
\medskip
\noindent\begin{tabular}{@{}p{14em}@{\quad}l@{}}
  \textsc{Parallel Superposition I:}
&
  $\displaystyle{\frac
  {(D' \lor {t \eq t'} \cdot \theta) \qquad (C[u,\dots,u]_{p_1,\dots,p_k} \cdot \theta)}
  {((D' \lor C[t',\dots,t']_{p_1,\dots,p_k})\sigma \cdot \theta)}
  }$
\end{tabular}\par\nobreak\noindent\begin{itemize}\item[]
  where $u$ is not a variable;
  $t\theta = u\theta$;
  $\sigma = \mgu(t {\doteq} u)$;
  $p_1,\dots,p_k$ are all the occurrences of $u$ in $C$;
  if one of the occurrences of $u$ in $C$ is in a negative literal
  or below the top in a positive literal then
  $(D' \lor {t \eq t'})\theta \precC C\theta$;
  one of the occurrences of $u$ in $C$ is either in a
  positive literal $s[{u}] \eq s'$ such that
  $(s[{u}] \eq s')\theta$
  is strictly maximal in $C\theta$
  or in a negative literal $s[{u}] \noteq s'$ such that
  $(s[{u}] \noteq s')\theta$
  is maximal in $C\theta$;
  $s[{u}]\theta \succ s'\theta$;
  $(t \eq t')\theta$ is strictly
  maximal in $(D' \lor {t \eq t'})\theta$; and $t\theta \succ t'\theta$.
\end{itemize}

\medskip
\noindent\begin{tabular}{@{}p{14em}@{\quad}l@{}}
  \textsc{Parallel Superposition II:}
&
  $\displaystyle{\frac
  {(D' \lor {t \eq t'} \cdot \theta) \qquad (C \cdot \theta)}
  {(D' \lor C \cdot \theta[x \mapsto u[t'\theta]])}
  }$
\end{tabular}\par\nobreak\noindent\begin{itemize}\item[]
  where $x$ is a variable of $C$;
  $x\theta = u[t\theta]$;
  if one of the occurrences of $x$ in $C$ is in a negative literal
  or below the top in a positive literal then
  $(D' \lor {t \eq t'})\theta \precC C\theta$;
  one of the occurrences of $x$ in $C$ is either in a
  positive literal ${s[{x}] \eq s'}$ such that
  ${(s[{x}] \eq s')\theta}$
  is strictly maximal in $C\theta$
  or in a negative literal $s[{x}] \noteq s'$ such that
  $(s[{x}] \eq s')\theta$
  is maximal in $C\theta$;
  $s[{x}]\theta \succ s'\theta$;
  $(t \eq t')\theta$ is strictly
  maximal in $(D' \lor {t \eq t'})\theta$; and $t\theta \succ t'\theta$.
\end{itemize}

\medskip
\noindent\begin{tabular}{@{}p{14em}@{\quad}l@{}}
  \textsc{Equality Resolution:}
&
  $\displaystyle{\frac
  {(C' \lor {s \noteq s'} \cdot \theta)}
  {(C'\sigma \cdot \theta)}
  }$
\end{tabular}\par\nobreak\noindent\begin{itemize}\item[]
  where $s\theta = s'\theta$;
  $\sigma = \mgu(s {\doteq} s')$;
  and
  $(s \noteq s')\theta$ is
  maximal in $(C' \lor {s \noteq s'})\theta$.
\end{itemize}

The following lemmas
compare the conclusion $\concl(\iota)$ of an inference $\iota$
with its right or only premise:

\begin{lem}
\label{lem:eq-res-reduces}%
Let $\iota$ be a ground \textsc{Equality Resolution} inference.
Then $\concl(\iota)$ is $\ngt_R$-smaller than its premise.
\end{lem}

\begin{proof}
Since every term that occurs in $\nm_R(C' \cdot \theta)$
occurs also in $\nm_R(C' \lor s \noteq s' \cdot \theta)$,
possibly with a larger label, the multiset
$\nm_R(C' \cdot \theta)$ is smaller than or equal to
$\nm_R(C' \lor s \noteq s' \cdot \theta)$, and
if the two multisets are equal,
then $C'\theta \precC (C' \lor s \noteq s')\theta$.
Therefore $(C' \cdot \theta) \nlt_R (C' \lor {s \noteq s'} \cdot \theta)$.
If $C'$ and $C'\sigma$ are equal up to bijective
renaming, we are done, otherwise the result follows from
Lemma~\ref{lem:ordering-of-instances}.
\qed
\end{proof}

\begin{lem}
\label{lem:par-sup-reduces}%
Let $\iota$ be a ground \textsc{Parallel Superposition} inference
\[\frac
  {(D' \lor {t \eq t'} \cdot \theta) \qquad (C[u,\dots,u]_{p_1,\dots,p_k} \cdot \theta)}
  {((D' \lor C[t',\dots,t']_{p_1,\dots,p_k})\sigma \cdot \theta)}
\]
with $t\theta = u\theta$ and $\sigma = \mgu(t {\doteq} u)$ or
\[\frac
  {(D' \lor {t \eq t'} \cdot \theta) \qquad (C \cdot \theta)}
  {(D' \lor C \cdot \theta[x \mapsto u[t'\theta]])}
\]
with $x\theta = u[t\theta]$.
If $(t\theta \to t'\theta) \in R$,
then $\concl(\iota)$ is $\ngt_R$-smaller than $(C\cdot\theta)$.
\end{lem}

\begin{proof}
Since $t\theta$ is replaced by $t'\theta$
at all occurrences of $u$ or at or below all occurrences of $x$ in $C$,
one copy of the redex $t\theta$ is removed
from $\nm_R(C\cdot\theta)$.
Moreover all terms in $D'\theta$ are smaller than $t\theta$,
and consequently all redexes stemming from $D'\theta$ are smaller than
$t\theta$.
Therefore $\nm_R(C\cdot\theta)$ is larger than
$\nm_R(D' \lor C[t',\dots,t']_{p_1,\dots,p_k} \cdot \theta)$
or
$\nm_R(D' \lor C \cdot \theta[x \mapsto u[t'\theta]])$.
In the second case, this implies
$(C\cdot\theta) \ngt_R \concl(\iota)$ immediately.
In the first case, it implies
$(C\cdot\theta) \ngt_R
(D' \lor C[t',\dots,t']_{p_1,\dots,p_k} \cdot \theta)$
and $(C\cdot\theta) \ngt_R \concl(\iota)$ follows
using Lemma~\ref{lem:ordering-of-instances}.
\qed
\end{proof}

\subsubsection{Redundancy.}
We will now construct a redundancy criterion for
ground closure Horn superposition that is based on the
ordering(s) $\ngt_R$.

\begin{dfn}
\label{dfn:redundancy-clo}%
Let $N$ be a set of ground closures.
A ground closure $(C\cdot\theta)$ is called redundant \wrt~$N$,
if for every left-reduced ground rewrite system $R$ contained in $\succ$
we have
(i) $R \models (C\cdot\theta)$ or
(ii) there exists a ground closure $(D\cdot\theta) \in N$
such that $(D\cdot\theta) \nlt_R (C\cdot\theta)$
and $R \not\models (D\cdot\theta)$.
\end{dfn}

\begin{dfn}
\label{dfn:redundancy-inf}%
Let $N$ be a set of ground closures.
A ground inference $\iota$ with right or only premise $(C\cdot\theta)$
is called redundant \wrt~$N$,
if for every left-reduced ground rewrite system $R$ contained in $\succ$
we have
(i) $R \models \concl(\iota)$, or
(ii) there exists a ground closure $(C'\cdot\theta) \in N$
such that $(C'\cdot\theta) \nlt_R (C\cdot\theta)$
and $R \not\models (C'\cdot\theta)$, or
(iii) $\iota$ is a \textsc{Superposition} inference
with left premise $(D' \lor t \eq t' \cdot \theta)$
where $t\theta \succ t'\theta$,
and $(t\theta \to t'\theta) \notin R$,
or (iv)~$\iota$~is a \textsc{Superposition} inference
where the left premise is not the $\ngt_R$-minimal premise.
\end{dfn}

Intuitively, a redundant closure cannot be a minimal
counterexample, \ie, a minimal closure that is false in $R$.
A redundant inference is either irrelevant for the
completeness proof (cases (iii) and (iv)),
or its conclusion (and thus its right or only premise) is
true in $R$, provided that all closures that are $\ngt_R$-smaller
than the right or only premise are true in $R$ (cases (i) and (ii))
-- which means that the inference can be used
to show that the right or only premise cannot be a minimal counterexample.

We denote the set of redundant closures
\wrt~$N$ by $\RedC(N)$
and the set of redundant inferences by $\RedI(N)$.

\begin{exam}
Let $\succ$ be a KBO where all symbols have weight $1$.
Let $C = {g(b) \noteq c} \lor {f(c) \noteq d}$
and $C' = f(g(b)) \noteq d$.
Then the closure $(C \cdot \emptyset)$
is redundant \wrt~${\{(C' \cdot \emptyset)\}}$:
Let $R$ be a left-reduced ground rewrite system contained in $\succ$.
Assume that $C$ is false in $R$.
Then $g(b)$ and $c$ have the same $R$-normal form.
Consequently, every redex or normal form in
$\nm_R(C' \cdot \emptyset)$
was already present in
$\nm_R(C \cdot \emptyset)$.
Moreover, the labeled normal form $(c{\downarrow}_R,2)$
that is present in $\nm_R(C \cdot \emptyset)$
is missing in $\nm_R(C' \cdot \emptyset)$.
Therefore $(C \cdot \emptyset)
\ngt_R (C' \cdot \emptyset)$.
Besides, if $(C \cdot \emptyset)$ is false in $R$,
then $(C' \cdot \emptyset)$ is false as well.

Note that $C \precC C'$, therefore
$C$ is not classically redundant \wrt~$\{C'\}$.
\end{exam}

\begin{lem}
\label{lem:is-red-crit}
$(\RedI,\RedC)$ is a redundancy criterion
in the sense of
(Waldmann et al.~\cite{WaldmannTourretRobillardBlanchette2022}),
that is,
(1)
  if $N \models \bot$,
  then $N \setminus \RedC(N) \models \bot$;
(2)
  if $N \subseteq N'$, then $\RedC(N) \subseteq \RedC(N')$
  and $\RedI(N) \subseteq \RedI(N')$;
(3)
  if $N' \subseteq \RedC(N)$,
  then $\RedC(N) \subseteq \RedC(N \setminus N')$ and
  $\RedI(N) \subseteq \RedI(N \setminus N')$; and
(4)
  if $\iota$ is an inference with conclusion in $N$,
  then $\iota \in \RedI(N)$.
\end{lem}

\begin{proof}
(1) Suppose that $N \setminus \RedC(N) \not\models \bot$.
Then there exists a left-reduced ground rewrite system $R$ contained in $\succ$
such that
$R \models N \setminus \RedC(N)$.
We show that $R \models N$ (which implies $N \not\models \bot$).
Assume that $R \not\models N$. Then there exists a closure
$(C\cdot\theta) \in N \cap \RedC(N)$
such that $R \not\models (C\cdot\theta)$.
By well-foundedness of $\ngt_R$ there exists a $\ngt_R$-minimal closure
$(C\cdot\theta)$ with this property.
By definition of $\RedC(N)$,
there must be a ground closure $(D\cdot\theta) \in N$
such that $(D\cdot\theta) \nlt_R (C\cdot\theta)$
and $R \not\models (D\cdot\theta)$.
By minimality of $(C\cdot\theta)$, we get
$(D\cdot\theta) \in N \setminus \RedC(N)$,
contradicting the initial assumption.

(2) Obvious.

(3) Let $N' \subseteq \RedC(N)$
and let $(C\cdot\theta) \in \RedC(N)$.
We show that $(C\cdot\theta) \in \RedC(N \setminus N')$.
Choose $R$ arbitrarily.
If $R \models (C\cdot\theta)$, we are done.
Otherwise there exists a ground closure $(D\cdot\theta) \in N$
such that $(D\cdot\theta) \nlt_R (C\cdot\theta)$
and $R \not\models (D\cdot\theta)$.
By well-foundedness of $\ngt_R$ there exists a $\ngt_R$-minimal closure
$(D\cdot\theta)$ with this property.
If $(D\cdot\theta)$ were contained in $N'$ and hence in $\RedC(N)$,
there would exist
a ground closure $(D'\cdot\theta) \in N$
such that $(D'\cdot\theta) \nlt_R (D\cdot\theta)$
and $R \not\models (D'\cdot\theta)$, contradicting minimality.
Therefore $(D\cdot\theta) \in N \setminus N'$ as required.
The second part of (3) is proved analogously.

(4) Let $\iota$ be an inference with $\concl(\iota) \in N$.
Choose $R$ arbitrarily.
We have to show that
$\iota$ satisfies part (i), (ii), (iii), or (iv) of
Def.~\ref{dfn:redundancy-inf}.
Assume that (i), (iii), and (iv) do not hold.
Then $R \not\models \concl(\iota)$,
and
by Lemmas~\ref{lem:eq-res-reduces} and \ref{lem:par-sup-reduces},
$\concl(\iota)$ is $\ngt_R$-smaller than the right or only premise of $\iota$,
therefore part (ii) is satisfied
if we take $\concl(\iota)$ as $(C'\cdot\theta)$.
\qed
\end{proof}

\subsubsection{Constructing a Candidate Interpretation.}

In usual completeness proofs for superposition-like calculi,
one constructs a candidate interpretation (a set of ground rewrite rules)
for a saturated set of
ground clauses by induction over the \emph{clause ordering}.
In our case, this is impossible since
the limit \emph{closure ordering} depends on the generated set of
rewrite rules itself.
We can still construct the candidate interpretation by induction over the
\emph{term ordering}, though:
Instead of inspecting ground closures one by one as in the
classical construction,
we inspect all ground closures $(C\cdot\theta)$
for which $C\theta$ contains the maximal term $s$
simultaneously,
and if for at least one of them the usual conditions for
productivity are satisfied,
we choose the $\ngt_{R_s}$-smallest one of these
to extend $R_s$.

Let $N$ be a set of ground closures.
For every ground term $s$ we define $R_s = \bigcup_{t \prec s} E_t$.
Furthermore we define $E_s = \{ s \to s' \}$,
if $(C \cdot \theta)$ is the $\ngt_{R_s}$-smallest
closure in $N$ such that
$C = C' \lor u \eq u'$,
$s = u\theta$ is a strictly maximal term in $C\theta$, occurs only in a
positive literal of $C\theta$, and is
irreducible \wrt~$R_s$,
$s' = u'\theta$,
$C\theta$ is false in $R_s$, and
$s \succ s'$,
provided that such a closure $(C \cdot \theta)$ exists.
We say that $(C \cdot \theta)$ \emph{produces} $s \to s'$.
If no such closure exists, we define $E_s = \emptyset$.
Finally, we define $R_* = \bigcup_t E_t$.

The following two lemmas are proved as usual:

\begin{lem}
\label{lem:modconstr-monotonic-1}%
Let $s$ be a ground term, let $(C\cdot\theta)$ be a closure.
If every term that occurs in negative literals of $C\theta$ is smaller
than $s$ and every term that occurs in positive literals of $C\theta$
is smaller than or equal to $s$, and if $R_s \models (C\cdot\theta)$,
then $R_* \models (C\cdot\theta)$.
\end{lem}

\begin{lem}
\label{lem:modconstr-monotonic-2}%
If a closure $(C' \lor u \eq u' \cdot \theta)$ produces
$u\theta \to u'\theta$,
then $R_* \models (C' \lor {u \eq u'} \cdot \theta)$
and $R_* \not\models (C' \cdot \theta)$.
\end{lem}

\begin{lem}
\label{lem:modconstr-orderings-agree}%
Let $(C_1\cdot\theta)$ and $(C_2\cdot\theta)$ be two closures.
If $s$ is a strictly maximal term and occurs only positively in both
$C_1\theta$ and $C_2\theta$,
then $(C_1\cdot\theta) \ngt_{R_s} (C_2\cdot\theta)$
if and only if
$(C_1\cdot\theta) \ngt_{R_*} (C_2\cdot\theta)$.
\end{lem}

\begin{proof}
If $E_s = \emptyset$, then
$\nm_{R_s}(C_i\cdot\theta) = \nm_{R_*}(C_i\cdot\theta)$ ($i \in \{1,2\}$),
therefore trivially
$(C_1\cdot\theta) \ngt_{R_s} (C_2\cdot\theta)$ if and only if
$(C_1\cdot\theta) \ngt_{R_*} (C_2\cdot\theta)$.
It remains to consider the case $E_s = \{s \to s'\}$.
Note that $s$ is then irreducible \wrt~$\ngt_{R_s}$,
that its proper subterms are irreducible \wrt~$\ngt_{R_*}$,
and that both $\nm_{R_s}(C_1\cdot\theta)$ and $\nm_{R_s}(C_2\cdot\theta)$
contain exactly one copy of $(s,0)$.
For $i \in \{1,2\}$ we obtain $\nm_{R_*}(C_i\cdot\theta)$
from $\nm_{R_s}(C_i\cdot\theta)$ by adding the labeled redexes that
occur in the $R_*$-normalization of $s$ and the
labeled $R_*$-normal form of $s$.
Since adding the same multiset to both
$\nm_{R_s}(C_1\cdot\theta)$ and $\nm_{R_s}(C_2\cdot\theta)$
does not change the ordering of these multisets,
the result follows.
\qed
\end{proof}

\begin{lem}
\label{lem:productive-Dtheta-smaller-Ctheta}%
Let $(D\cdot\theta) = (D' \lor {t \eq t'} \cdot \theta)$
and $(C\cdot\theta)$ be two closures in $N$.
If $(D\cdot\theta)$ produces $t\theta \to t'\theta$ in $R_*$,
and $t\theta$ occurs in $C\theta$
in a negative literal
or below the top a term in a positive literal,
then $(D\cdot\theta) \nlt_{R_*} (C\cdot\theta)$
and $D\theta \prec_C C\theta$.
\end{lem}

\begin{proof}
If $(D\cdot\theta)$ produces $t\theta \to t'\theta$ in $R_*$,
then $t\theta$ is the strictly largest term in $D\theta$ and occurs
only in a positive literal.
Therefore the labeled redex $(t\theta,0)$ must be the largest element of
$\nm_{R_*}(D\cdot\theta)$.

If $t\theta$ occurs in $C\theta$ in a negative literal or
below the top of a positive literal,
then $\nm_{R_*}(C\cdot\theta)$ contains $(t\theta,2)$ or $(t\theta,1)$,
hence it is larger than $\nm_{R_*}(D\cdot\theta)$.
Consequently, $(D\cdot\theta) \nlt_{R_*} (C\cdot\theta)$.
Besides, the literal in which $t\theta$ occurs in $C\theta$
is larger than $t\theta \eq t'\theta$ \wrt~$\succL$,
and since $t\theta \eq t'\theta$ is 
larger than every literal of $D'\theta$,
we obtain $D\theta \prec_C C\theta$.
\qed
\end{proof}

\begin{lem}
\label{lem:productive-Dtheta-smaller-Ctheta-top-pos}%
Let $(D\cdot\theta) = (D' \lor {t \eq t'} \cdot \theta)$
and $(C\cdot\theta)$ be two closures in $N$.
If $(D\cdot\theta)$ produces $t\theta \to t'\theta$ in $R_*$,
$t\theta$ occurs in $C\theta$
at the top of the strictly maximal side of a positive maximal literal,
and $R_* \not\models (C\cdot\theta)$,
then $(D\cdot\theta) \nlt_{R_*} (C\cdot\theta)$.
\end{lem}

\begin{proof}
If $(D\cdot\theta)$ produces $t\theta \to t'\theta$ in $R_*$,
then $t\theta$ is the strictly largest term in $D\theta$ and occurs
only in a positive literal.
If $t\theta$ occurs in $C\theta$ only at the top
of the strictly maximal side of
a positive maximal literal, then $\ngt_{R_s}$ and $\ngt_{R_*}$
agree on $(D\cdot\theta)$ and $(C\cdot\theta)$
by Lemma~\ref{lem:modconstr-orderings-agree};
furthermore $R_* \not\models (C\cdot\theta)$ implies
$R_s \not\models (C\cdot\theta)$.
Then $(D\cdot\theta) \nlt_{R_*} (C\cdot\theta)$
folllows from the $\ngt_{R_s}$-minimality condition for
productive closures.
\qed
\end{proof}

We can now show that the Ground Closure Horn Superposition Calculus
is refutationally complete: 

\begin{thm}
\label{thm:rstar-is-model}%
Let $N$ be a saturated set of ground closures
that does not contain $(\bot\cdot\theta)$.
Then $R_* \models N$.
\end{thm}

\begin{proof}
Suppose that $R_* \not\models N$.
Let $(C\cdot\theta)$ be the $\ngt_{R_*}$-smallest closure in $N$
such that
$R_* \not\models (C\cdot\theta)$.

Case 1: $C = C' \lor s \noteq s'$
and $s\theta \noteq s'\theta$ is maximal in $C\theta$.
By assumption, $R_* \not\models s\theta \noteq s'\theta$,
hence $s\theta{\downarrow}_{R_*} = s'\theta{\downarrow}_{R_*}$.

Case 1.1: $s\theta = s'\theta$.
Then there is an \textsc{Equality Resolution} inference
from $(C\cdot\theta)$ with conclusion $(C'\sigma\cdot\theta)$,
where $\theta \circ \sigma = \theta$.
By saturation the inference is redundant,
and by minimality of $(C\cdot\theta)$ \wrt~$\ngt_{R_*}$ this implies
$R_* \models (C'\sigma\cdot\theta)$.
But then $R_* \models (C\cdot\theta)$,
contradicting the assumption.

Case 1.2: $s\theta \neq s'\theta$.
W.l.o.g.~let $s\theta \succ s'\theta$.
Then $s\theta$ must be reducible by a rule $t\theta \to t'\theta \in R_*$,
which has been produced by a closure
$(D \cdot \theta) = (D' \lor t \eq t' \cdot \theta)$ in $N$.
By Lemma~\ref{lem:productive-Dtheta-smaller-Ctheta},
$(D \cdot \theta) \nlt_{R_*} (C \cdot \theta)$
and $D\theta \precC C\theta$.
If $s\theta$ and $t\theta$ overlap at a non-variable position of $s$,
there is a \textsc{Parallel Superposition I} inference $\iota$ between
$(D\cdot\theta)$ and $(C\cdot\theta)$;
otherwise they overlap at or below a variable position of $s$
and there is a \textsc{Parallel Superposition II} inference $\iota$
with premises
$(D\cdot\theta)$ and $(C\cdot\theta)$.
By Lemma~\ref{lem:modconstr-monotonic-2}, $R_* \not\models (D'\cdot\theta)$.
By saturation the inference is redundant,
and by minimality of $(C\cdot\theta)$ \wrt~$\ngt_{R_*}$
we know that
$R_* \models \concl(\iota)$.
Since $R_* \not\models (D'\cdot\theta)$ this implies
$R_* \models (C\cdot\theta)$,
contradicting the assumption.

Case 2: $C\theta = C'\theta \lor s\theta \eq s'\theta$
and $s\theta \eq s'\theta$ is maximal in $C\theta$.
By assumption, $R_* \not\models s\theta \eq s'\theta$,
hence $s\theta{\downarrow}_{R_*} \neq s'\theta{\downarrow}_{R_*}$.
W.l.o.g.~let $s\theta \succ s'\theta$.

Case 2.1: $s\theta$ is reducible by a rule $t\theta \to t'\theta \in R_*$,
which has been produced by a closure
$(D \cdot \theta) = (D' \lor t \eq t' \cdot \theta)$ in $N$.
By Lemmas \ref{lem:productive-Dtheta-smaller-Ctheta}
and \ref{lem:productive-Dtheta-smaller-Ctheta-top-pos},
we obtain $(D \cdot \theta) \nlt_{R_*} (C \cdot \theta)$,
and, provided that $t\theta$ occurs in $s\theta$ below the top,
also $D\theta \precC C\theta$.
Therefore
there is a \textsc{Parallel Superposition} (I or II) inference $\iota$
with left premise $(D\cdot\theta)$ and right premise $(C\cdot\theta)$,
and we can derive a contradiction analogously to Case~1.2.

Case 2.2: It remains to consider the case that
$s\theta$ is irreducible by $R_*$.
Then $s\theta$ is also irreducible by $R_{s\theta}$.
Furthermore, by Lemma~\ref{lem:modconstr-orderings-agree},
$\ngt_{R_{s\theta}}$ and $\ngt_{R_*}$ agree on all closures
in which $s\theta$ is a strictly maximal term and occurs only positively.
Therefore $(C\cdot\theta)$ satisfies all conditions
for productivity, hence $R_* \models (C\cdot\theta)$,
contradicting the assumption.
\qed
\end{proof}

\subsection{Lifting}
\label{ssect:lifting}

It remains to lift the refutational completeness result
for ground closure Horn superposition to the non-ground case.

If $C$ is a general clause, we call every ground closure
$(C \cdot \theta)$ a ground instance of $C$.
If
\[\frac{C_n \ldots C_1}{C_0}\]
is a general inference
and
\[\frac{(C_n \cdot\theta) \ldots (C_1 \cdot \theta)}{(C_0 \cdot \theta)}\]
is a ground inference,
we call the latter a ground instance of the former.
The function $\gnd$ maps every general clause $C$
and every general inference $\iota$ to the set of its 
ground instances.
We extend $\gnd$ to sets of clauses $N$ or sets of inferences $I$
by defining $\gnd(N) := \bigcup_{C \in N} \gnd(C)$
and $\gnd(I) := \bigcup_{\iota \in I} \gnd(\iota)$.

\begin{lem}
\label{lem:grounding-function}%
$\gnd$ is a grounding function, that is,
(1) $\gnd(\bot) = \{\bot\}$;
(2) if $\bot \in \gnd(C)$ then $C = \bot$; and
(3) for every inference $\iota$,
$\gnd(\iota) \subseteq \RedI(\gnd(\concl(\iota)))$.
\end{lem}

\begin{proof}
Properties (1) and (2) are obvious;
property (3) follows from the fact that
for every inference $\iota$,
$\concl(\gnd(\iota)) \subseteq \gnd(\concl(\iota))$.
\qed
\end{proof}

The grounding function $\gnd$ induces a lifted redundancy criterion
$(\RedI^\gnd,\RedC^\gnd)$ where
$\iota \in \RedI^\gnd(N)$ if and only if
$\gnd(\iota) \subseteq \RedI(\gnd(N))$
and
$C \in \RedC^\gnd(N)$ if and only if
$\gnd(C) \subseteq \RedC(\gnd(N))$.

\begin{lem}
\label{lem:lifting}
Every ground inference from closures in $\gnd(N)$ is
a ground instance of an inference from $N$ or
contained in $\RedI(\gnd(N))$.
\end{lem}

\begin{proof}
Let $\iota$ be a ground inference from closures in $\gnd(N)$.
If $\iota$ is a \textsc{Parallel Superposition I} inference
\[\frac
  {(D \cdot \theta) \qquad (C \cdot \theta)}
  {(C_0\sigma \cdot \theta)}
\]
then it is a ground instance of the \textsc{Parallel Superposition} inference
\[\frac
  {D \qquad C}
  {C_0\sigma}
\]
whose premises are in $N$.
Similarly, if $\iota$ is an \textsc{Equality Resolution} inference,
then it is a ground instance of an \textsc{Equality Resolution}
with premises in $N$.

Otherwise
$\iota$ is a \textsc{Parallel Superposition II} inference
\[\frac
  {(D' \lor {t \eq t'} \cdot \theta) \qquad (C \cdot \theta)}
  {(D' \lor C \cdot \theta[x \mapsto u[t'\theta]])}
\]
with $x\theta = u[t\theta]$.
Let $R$ be a left-reduced ground rewrite system contained in $\succ$.
If $(t\theta \to t'\theta) \notin R$, then
$\iota$ satisfies case (iii) of Def.~\ref{dfn:redundancy-inf}.
Otherwise $(C \cdot \theta[x \mapsto u[t'\theta]])$
is a ground instance of the clause $C \in N$ and $\ngt_R$-smaller than the
premise $(C \cdot \theta)$.
If $R \models (C \cdot \theta[x \mapsto u[t'\theta]])$,
then $R \models \concl(\iota)$,
so $\iota$ satisfies case (i) of Def.~\ref{dfn:redundancy-inf};
otherwise it satisfies case (ii) of Def.~\ref{dfn:redundancy-inf}.
\qed
\end{proof}

\begin{thm}
The Horn Superposition Calculus (using \textsc{Parallel Superposition})
together with the lifted redundancy criterion
$(\RedI^\gnd,\RedC^\gnd)$ is refutationally complete.
\end{thm}

\begin{proof}
This follows immediately from Lemma~\ref{lem:lifting}
and Thm.~32 from
(Waldmann et al.~\cite{WaldmannTourretRobillardBlanchette2022}).
\qed
\end{proof}

\subsection{Deletion and Simplification}

It turns out that \DER{}, as well as
most concrete deletion and simplification techniques
that are implemented in state-of-the-art superposition provers
are in fact covered by our abstract redundancy criterion.
There are some unfortunate exceptions, however.

\subsubsection{DER.}
\textsc{Destructive Equality Resolution}, that is, the replacement
of a clause $x \noteq t \lor C$ with $x \notin \vars(t)$ by
$C\{x \mapsto t\}$ is covered by the redundancy criterion.
To see this, consider an arbitrary ground instance
$(x \noteq t \lor C \cdot \theta)$ of
$x \noteq t \lor C$.
Let $R$ be a left-reduced ground rewrite system contained in $\succ$.
Assume that the instance is false in $R$.
Then $x\theta$ and $t\theta$ have the same $R$-normal form.
Consequently, any redex or normal form in
$\nm_R(C\{x \mapsto t\} \cdot \theta)$
was already present in
$\nm_R(x \noteq t \lor C \cdot \theta)$
(possibly with a larger label, if $x$ occurs only positively in $C$).
Moreover, the labeled normal form $(x\theta{\downarrow}_R,2)$
that is present in $\nm_R(x \noteq t \lor C \cdot \theta)$
is missing in $\nm_R(C\{x \mapsto t\} \cdot \theta)$.
Therefore $(x \noteq t \lor C \cdot \theta)
\ngt_R (C\{x \mapsto t\} \cdot \theta)$.
Besides, both closures have clearly the same truth value in $R$,
that is, false.

\subsubsection{Subsumption.}
Propositional subsumption, that is,
the deletion of a clause ${C \lor D}$ with nonempty $D$
in the presence of a clause $C$
is covered by the redundancy criterion.
This follows from the fact that every
ground instance
$((C \lor D) \cdot \theta)$ of the deleted clause
is entailed by a smaller ground instance
$(C \cdot \theta)$ of the subsuming clause.
This extends to all simplifications that replace a clause
by a subsuming clause in the presence of certain other clauses,
for instance the replacement of a clause
$t\sigma \eq t'\sigma \lor C$ by $C$ in the presence
of a clause $t \noteq t'$,
or the replacement of a clause
$u[t\sigma] \noteq u[t'\sigma] \lor C$ by $C$ in the presence
\looM
of a clause $t \eq t'$.

First-order subsumption, that is,
the deletion of a clause $C\sigma \lor D$
in the presence of a clause $C$
is not covered, however.
This is due to the fact that
$\ngt_R$ makes the instance $C\sigma$ smaller than $C$,
rather than larger (see Lemma~\ref{lem:ordering-of-instances}).

\subsubsection{Tautology Deletion.}
The deletion of (semantic or syntactic) tautologies is obviously
covered by the redundancy criterion.

\subsubsection{Parallel Rewriting with Condition Literals.}
Parallel rewriting with condition literals, that is, the replacement
of a clause $t \noteq t' \lor C[t,\ldots,t]_{p_1,\dots,p_k}$,
where $t \succ t'$ and $p_1,\dots,p_k$ are all the occurrences of $t$ in $C$
by $t \noteq t' \lor C[t',\ldots,t]_{p_1,\dots,p_k}$
is covered by the redundancy criterion.
This can be shown analogously as for~\DER{}.

\subsubsection{Demodulation.}
Parallel demodulation is the replacement of a clause \linebreak[4]
$C[t\sigma,\ldots,t\sigma]_{p_1,\dots,p_k}$
by $C[t'\sigma,\ldots,t'\sigma]_{p_1,\dots,p_k}$
in the presence of another clause $t \eq t'$ where $t\sigma \succ t'\sigma$.
In general, this is \emph{not} covered by our redundancy criterion.
For instance, if $\succ$ is a KBO where all symbols have weight $1$
and if $R = \{f(b) \to b,\,g(g(b)) \to b\}$,
then replacing $f(f(f(b)))$ by $g(g(b))$ in some clause
$f(f(f(b))) \noteq c$ yields a clause with a larger $R$-normalization
multiset,
since the labeled redexes $\{(f(b),2),\,(f(b),2),\,(f(b),2)\}$ are replaced
by $\{(g(g(b)),2)\}$
and $g(g(b)) \succ f(b)$.

A special case is supported, though: If $t'\sigma$ is a proper subterm of
$t\sigma$, then the $R$-normalization multiset either remains the same
of becomes smaller, since every redex in the normalization of
$t'\sigma$ occurs also in the normalization of
$t\sigma$.

\section{Completeness, Part II: The Non-Horn Case}
\label{sect:complnh}%

In the non-Horn case, the construction that we have seen in the
previous section fails for \textsc{(Parallel) Superposition} inferences
at the top of positive literals.
Take an LPO with precedence
$f \succ c_6 \succ c_5 \succ c_4 \succ c_3 \succ c_2 \succ c_1 \succ b$
and consider the ground closures
$(f(x_1) \eq c_1 \lor f(x_2) \eq c_2 \lor f(x_3) \eq c_3 \cdot \theta)$
and
$({f(x_4) \eq c_4} \lor {f(x_5) \eq c_5} \lor {f(x_6) \eq c_6} \cdot \theta)$,
where $\theta$ maps all variables to the same constant $b$.
Assume that the first closure produces the rewrite rule
$(f(b) \eq c_3) \in R_*$.
The $R_*$-normalization multisets of both closures
are dominated by three occurrences of the labeled redex $(f(b),0)$.
However, a \textsc{Superposition} inference between the closures
yields
$(f(x_1) \eq c_1 \lor f(x_2) \eq c_2 \lor
f(x_4) \eq c_4 \lor f(x_5) \eq c_5 \lor c_3 \eq c_6 \cdot \theta)$,
whose $R_*$-normalization multiset
contains four occurrences of the labeled redex $(f(b),0)$,
hence the conclusion of the inference is larger than
both premises.
If we want to change this,
we must ensure that the weight of positive literals
depends primarily on their larger sides,
and if the larger sides are equal,
on their smaller sides.
That means that in the non-Horn case,
the clause ordering must treat positive literals as the
traditional clause ordering $\succC$.
But that has two important consequences:
First, \DER{} may no longer be used to eliminate variables that
occur also in positive literals
(since \DER{} might now increase the weight of these literals).
On the other hand, unrestricted demodulation becomes possible
for positive literals.

\begin{dfn}
\label{dfn:ss-ts-nh}%
We define the subterm set $\ulss^{-}(C)$
and the topterm set $\ults^{-}(C)$ of a clause $C$ as in the Horn case:
\[\begin{array}{@{}r@{}l@{}}
  \ulss^{-}(C) = {} & \{\, t \mid C = C' \lor s[t]_p \noteq s' \,\} \\[1ex]
  \ults^{-}(C) = {} & \{\, t \mid C = C' \lor t \noteq t' \,\}
\end{array}\]
\end{dfn}

We do not need labels anymore.
Instead, for every redex or normal form $u$ that appears in negative
literals we include the two-element multiset $\{u.u\}$
in the $R$-normalization multiset
to ensure that a redex or normal form $u$
in negative literals has a larger weight than a positive
literal $u \eq v$ with $u \succ v$:

\begin{dfn}
\label{dfn:rm-nm-ngt-nh}%
We define the $R$-redex multiset $\rrm_R(t)$
of a ground term $t$ by
\[\begin{array}{@{}r@{}l@{}}
  \rrm_R(t) = {} & \emptyset \text{~if $t$ is $R$-irreducible;} \\[1ex]
  \rrm_R(t) = {} & \{\{u,u\}\} \cup \rrm_R(t') \text{~if $t \to_R t'$ using the rule $u \to v \in R$.}
\end{array}\]
We define the $R$-normalization multiset $\nm_R(C\cdot\theta)$
of a ground closure $(C\cdot\theta)$ by
\[\begin{array}{@{}r@{}l@{}}
  \nm_R(C\cdot\theta) = {} & \,\bigcup_{f(t_1,\dots,t_n) \in \ulss^{-}(C)}
         \rrm_R(f(t_1\theta{\downarrow}_R,\dots,t_n\theta{\downarrow}_R)) \\[3pt]
 & \hspace*{2em} {} \cup \bigcup_{x \in \ulss^{-}(C)} \rrm_R(x\theta) \\[3pt]
 & \hspace*{2em} {} \cup \bigcup_{t \in \ults^{-}(C)} \{ \{t\theta{\downarrow}_R,t\theta{\downarrow}_R\} \} \\[3pt]
 & \hspace*{2em} {} \cup \bigcup_{(s \eq s') \in C} \{ \{s\theta,s'\theta\} \}
\end{array}\]
Once more, the $R$-normalization closure ordering $\ngt_R$ compares
ground closures $(C\cdot\theta_1)$ and $(D\cdot\theta_2)$
using a lexicographic combination of three orderings:
\begin{itemize}
\item
  first, the twofold multiset extension $({\succ}_\mul)_\mul$
  of the reduction ordering $\succ$
  applied to the multisets $\nm_R(C\cdot\theta_1)$ and $\nm_R(D\cdot\theta_2)$,
\item
  second, the traditional clause ordering $\succC$ applied to
  $C\theta_1$ and $D\theta_2$,
\item
  and third, the same closure ordering $\succCL$ as in the
  Horn case.
\end{itemize}
\end{dfn}

Lemma~\ref{lem:ordering-of-instances-nh} can be proved like
Lemma~\ref{lem:ordering-of-instances}:

\begin{lem}
\label{lem:ordering-of-instances-nh}%
If $(C\cdot\theta)$ and $(C\sigma\cdot\theta')$
are ground closures,
such that
$C\theta = C\sigma\theta'$,
and $C$ and $C\sigma$ are not equal up to bijective renaming,
then $(C\cdot\theta) \ngt_R (C\sigma\cdot\theta')$.
\end{lem}

The ground closure superposition calculus for the non-Horn case
uses the inference rules \textsc{Parallel Superposition I},
\textsc{Parallel Superposition II}, \textsc{Equality Resolution},
defined as in the Horn case,
plus the \textsc{Equality Factoring} rule:\footnote{%
  Alternatively, one can replace \textsc{Equality Factoring} by
  \textsc{Ordered Factoring} and two variants of
  \textsc{Merging Paramodulation},
  one for overlaps at a non-variable position and
  one for overlaps at or below a variable position
  (analogously to \textsc{Parallel Superposition}).}

\bigskip
\noindent\begin{tabular}{@{}p{14em}@{~}l@{}}
  \textsc{Equality Factoring:}
&
  $\displaystyle{\frac
  {(C' \lor {r} \eq r' \lor {s} \eq s' \cdot \theta)}
  {((C' \lor s' \noteq r' \lor r \eq r')\sigma \cdot \theta)}
  }$
\end{tabular}\par\nobreak\noindent\begin{itemize}\item[]
  where $s\theta = r\theta$;
  $\sigma = \mgu(s {\doteq} r)$;
  $s\theta \succ s'\theta$;
  and
  $({s} \eq s')\theta$ is maximal in
  $(C' \lor {r} \eq r' \lor {s} \eq s')\theta$.
\end{itemize}

We can prove the equivalent versions of
Lemmas~\ref{lem:eq-res-reduces} and \ref{lem:par-sup-reduces}
and their analogue for \textsc{Equality Factoring}:

\begin{lem}
\label{lem:eq-res-reduces-nh}%
Let $\iota$ be a ground \textsc{Equality Resolution} inference.
Then the conclusion $\concl(\iota)$ of $\iota$
is $\ngt_R$-smaller than its premise.
\end{lem}

\begin{proof}
The multiset
$\nm_R(C' \cdot \theta)$ is a submultiset of
$\nm_R(C' \lor s \noteq s' \cdot \theta)$
and $C'\theta \precC (C' \lor s \noteq s')\theta$.
Therefore $(C' \cdot \theta) \nlt_R (C' \lor {s \noteq s'} \cdot \theta)$.
If $C'$ and $C'\sigma$ are equal up to bijective
renaming, we are done, otherwise the result follows from
Lemma~\ref{lem:ordering-of-instances}.
\qed
\end{proof}

\begin{lem}
\label{lem:par-sup-reduces-nh}%
Let $\iota$ be a ground \textsc{Parallel Superposition} inference
\[\frac
  {(D' \lor {t \eq t'} \cdot \theta) \qquad (C[u,\dots,u]_{p_1,\dots,p_k} \cdot \theta)}
  {((D' \lor C[t',\dots,t']_{p_1,\dots,p_k})\sigma \cdot \theta)}
\]
with $t\theta = u\theta$ and $\sigma = \mgu(t {\doteq} u)$ or
\[\frac
  {(D' \lor {t \eq t'} \cdot \theta) \qquad (C \cdot \theta)}
  {(D' \lor C \cdot \theta[x \mapsto u[t'\theta]])}
\]
with $x\theta = u[t\theta]$.
If $(D' \lor {t \eq t'} \cdot \theta) \nlt_R (C\cdot\theta)$
and $(t\theta \to t'\theta) \in R$,
then $\concl(\iota)$ is $\ngt_R$-smaller than the premise $(C\cdot\theta)$.
\end{lem}

\begin{proof}
The redex $t\theta$ is replaced by $t'\theta$
at all occurrences of $u$ or at or below all occurrences of $x$ in $C$.

(1) If at least one of the occurrences where $u$ or $x$ occurs in $C$
is in a negative literal, then one redex $\{t\theta,t\theta\}$ is removed
from $\nm_R(C\cdot\theta)$,
furthermore some of the multisets $\{s\theta,s'\theta\}$
in $\nm_R(C\cdot\theta)$ that stem from positive literals
may be replaced by smaller ones as well.
The multiset $\{t\theta,t'\theta\}$ is a strictly maximal element of
$\nm_R(D' \lor {t \eq t'} \cdot \theta)$;
therefore all multisets in $\nm_R(D\cdot\theta)$
are smaller than $\{t\theta,t'\theta\}$ and thus smaller than
$\{t\theta,t\theta\}$.
Therefore $\nm_R(C\cdot\theta)$ is larger than
$\nm_R(D' \lor C[t',\dots,t']_{p_1,\dots,p_k} \cdot \theta)$
or
$\nm_R(D' \lor C \cdot \theta[x \mapsto u[t'\theta]])$.
In the second case, this implies
$(C\cdot\theta) \ngt_R \concl(\iota)$ immediately.
In the first case, it implies
$(C\cdot\theta) \ngt_R
(D' \lor C[t',\dots,t']_{p_1,\dots,p_k} \cdot \theta)$
and $(C\cdot\theta) \ngt_R \concl(\iota)$ follows
using Lemma~\ref{lem:ordering-of-instances}.

(2) If $u$ or $x$ occurs only in positive literals of $C$,
then there is a strictly maximal literal
$s\theta \eq s'\theta$ of $C\theta$
such that $t\theta$ is a subterm of $s\theta$.
Then the multiset $\{s\theta,s'\theta\}$
in $\nm_R(C\cdot\theta)$ that stems from the literal
$s \eq s'$
is replaced by a smaller multiset.
If $t\theta$ occurs in $s\theta$ below the top, then
$t\theta \prec s\theta$, therefore
$\{t\theta,t'\theta\}$ is smaller than
$\{s\theta,s'\theta\}$, therefore all multisets in $\nm_R(D\cdot\theta)$
are smaller than $\{s\theta,s'\theta\}$,
and we can continue as in part (1).
Otherwise $t\theta = s\theta$, then we must have $t'\theta \preceq s'\theta$,
since otherwise $(D' \lor {t \eq t'} \cdot \theta) \ngt_R (C\cdot\theta)$.
So we can again conclude that
all multisets in $\nm_R(D\cdot\theta)$ are smaller than
$\{t\theta,t'\theta\}$, which is smaller than or equal to
$\{s\theta,s'\theta\}$, and we can again
continue as in part (1).
\qed
\end{proof}

\begin{lem}
\label{lem:eq-fact-reduces-nh}%
Let $\iota$ be a ground \textsc{Equality Factoring} inference
with premise $(C \cdot \theta)$.
Then $\concl(\iota)$ is $\ngt_R$-smaller than
$(C \cdot \theta)$.
\end{lem}

\begin{proof}
We have $C = C' \lor {r} \eq r' \lor {s} \eq s'$
and
$\concl(\iota) = ((C' \lor s' \noteq r' \lor r \eq r')\sigma \cdot \theta)$,
where $s\theta = r\theta$,
$\sigma = \mgu(s {\doteq} r)$,
$s\theta \succ s'\theta$,
and $s\theta \eq s'\theta$ is a maximal literal in $C\theta$.
From this we can conclude that
$r'\theta \preceq s'\theta$.
In the inference, the multiset $\{s\theta,s'\theta\}$
in $\nm_R(C\cdot\theta)$ that stems from the literal
$s \eq s'$
is removed;
instead multisets of the form
$\{u,u\}$ that stem from the redexes and normal forms
of $s'\theta$ and $r'\theta$ may be added.
Since both $s'\theta$ and $r'\theta$ are smaller than $s\theta$,
every term $u$ is also smaller than $s\theta$.
Therefore all the potentially added multisets are smaller than the
removed one, thus
$(C \cdot \theta) \ngt_R (C' \lor s' \noteq r' \lor r \eq r' \cdot \theta)$.
This implies
$(C \cdot \theta) \ngt_R \concl(\iota)$
by Lemma~\ref{lem:ordering-of-instances}.
\qed
\end{proof}

Redundant closures and inferences are defined in the same way as in
Def. \ref{dfn:redundancy-clo} and~\ref{dfn:redundancy-inf}.\footnote{%
  If one uses \textsc{Merging Paramodulation} instead of
  \textsc{Equality Factoring},
  condition (iii) of Def.~\ref{dfn:redundancy-inf}
  must hold also for \textsc{Merging Paramodulation} inferences.}

In the construction of a candidate interpretation,
we define again $R_s = \bigcup_{t \prec s} E_t$
for every ground term $s$
and $R_* = \bigcup_t E_t$.
We define $E_s = \{ s \to s' \}$,
if $(C \cdot \theta)$ is the $\ngt_{R_s}$-smallest
closure in $N$ such that
$C = C' \lor u \eq u'$,
$u\theta \eq u'\theta$ is a strictly maximal literal in $C\theta$,
$s = u\theta$, $s' = u'\theta$,
$s \succ s'$,
$s$ is irreducible \wrt~$R_s$,
$C\theta$ is false in $R_s$, and
$C'\theta$ is false in $R_s \cup \{s \to s'\}$,
provided that such a closure $(C \cdot \theta)$ exists.
If no such closure exists, we define $E_s = \emptyset$.

With these definitions,
Lemmas~\ref{lem:modconstr-monotonic-1}--\ref{lem:productive-Dtheta-smaller-Ctheta-top-pos}
continue to hold.
We can then reprove Thm.~\ref{thm:rstar-is-model}
for the non-Horn case.
The only difference in the proof is one additional subcase
before Case 2.1:
If $s\theta \eq s'\theta$ is maximal, but not strictly maximal
in $C\theta$,
or if $C'\theta$ is true in $R_{s\theta} \cup \{s\theta \to s'\theta\}$,
then there is an \textsc{Equality Factoring}
inference with the premise $(C\cdot\theta)$.
This inference must be redundant, which yields again a contradiction.

The lifting to non-ground clauses works as in
Sect.~\ref{ssect:lifting}.

\section{Discussion}
\label{sect:disc}%

We have demonstrated that the naive addition of
\textsc{Destructive Equality Resolution} (\DER{})
to the standard abstract redundancy
concept destroys the refutational completeness of the calculus,
but that there exist restricted variants of
the Superposition Calculus that are refutationally complete even with
\DER{}
(restricted to negative literals in the non-Horn case).
The key tool for the completeness proofs is
a closure ordering that is structurally very different
from the classical ones~-- it is not a multiset extension of some
literal ordering~-- but that still has the property that
the redundancy criterion induced by it is compatible with
the Superposition Calculus.

Of course there is a big gap between the negative result and
the positive results. The new redundancy criterion justifies
\DER{} as well as most deletion and
simplification techniques found in realistic saturation provers,
but only propositional subsumption
and only a very restricted variant of demodulation,
The question whether the Superposition Calculus is
refutationally complete together with a redundancy criterion
that justifies both \DER{}
(in full generality even in the non-Horn case)
and \emph{all} deletion and
simplification techniques found in realistic saturation provers
(including unrestricted demodulation and first-order subsumption)
is still open.
Our work is an intermediate step towards a solution to this problem.
There may exist a more refined closure ordering that
allows us to prove the completeness of such a calculus.
On the other hand, if the combination is really incomplete,
a counterexample must make use of those operations
that our proof fails to handle, that is,
\DER{} in positive literals in non-Horn problems,
first-order subsumption, or
demodulation with unit equations that are contained in the
usual term ordering $\succ$ but yield closures that are
larger \wrt~$\ngt_{R_*}$.

\subsubsection{Acknowledgments.}

I thank the anonymous \acr{IJCAR} reviewers for their helpful comments
and Stephan Schulz for drawing my attention to the problem
at the \acr{CASC} dinner in 2013.

\bibliography{paper}

\end{document}